\documentclass[11pt]{article}

\usepackage[colorlinks,pagebackref]{hyperref}
\usepackage{amsmath} 
\usepackage{amsthm} 
\usepackage{amssymb}	
\usepackage{graphicx} 
\usepackage{multicol} 
\usepackage{multirow}
\usepackage{color}
\usepackage[dvips,letterpaper,margin=1in,bottom=1in]{geometry}
\usepackage[capitalize,noabbrev]{cleveref}

\usepackage[utf8]{inputenc}
\usepackage[english]{babel}

\usepackage{mathtools}

\newtheorem{theorem}{Theorem}[section]

\newtheorem{lemma}[theorem]{Lemma}
\newtheorem{corollary}[theorem]{Corollary}

\DeclarePairedDelimiter\rbra{\lparen}{\rparen}
\DeclarePairedDelimiter\sbra{\lbrack}{\rbrack}
\DeclarePairedDelimiter\cbra{\{}{\}}

\DeclarePairedDelimiter\Abs{\lVert}{\rVert}

\DeclarePairedDelimiter\floor{\lfloor}{\rfloor}
\DeclarePairedDelimiter\ket{\lvert}{\rangle}
\DeclarePairedDelimiter\bra{\langle}{\rvert}
\DeclarePairedDelimiter\ave{\langle}{\rangle}

\newcommand{\substr}[2] {\sbra*{#1 .. #2}}

\DeclareMathOperator*{\argmin}{arg\,min}

\newcommand{\set}[2] {\left\{\, #1 \colon #2 \,\right\}}

\newcommand{\polylog} {\operatorname{polylog}}

\usepackage{algorithm}
\usepackage{algpseudocode}

\usepackage{tabularx}
\usepackage{booktabs}
\usepackage{threeparttable}

\usepackage{adjustbox}

\newcommand{\footremember}[2]{%
    \footnote{#2}
    \newcounter{#1}
    \setcounter{#1}{\value{footnote}}%
}

\usepackage{tikz}
\usetikzlibrary{quantikz2}

\begin{document}

\title{Quantum Data Structure for Range Minimum Query}

\author{
    Qisheng Wang \footremember{1}{Qisheng Wang is with the School of Informatics, University of Edinburgh, United Kingdom (e-mail: \url{QishengWang1994@gmail.com}).}
    \and
    Zhean Xu \footremember{2}{Zhean Xu is with the Department of Computer Science and Technology, Tsinghua University, China (e-mail: \url{xuzhean@icloud.com}).}
    \and Zhicheng Zhang \footremember{3}{Zhicheng Zhang is with the Centre for Quantum Software and Information, University of Technology Sydney, Australia (e-mail:\url{iszczhang@gmail.com}).}
}
\date{}

\maketitle

\begin{abstract}
Given an array $a\substr{1}{n}$, the Range Minimum Query (RMQ) problem is to maintain a data structure that supports RMQ queries: given a range $\sbra{l, r}$, find the index of the minimum element among $a\substr{l}{r}$, i.e., $\argmin_{i \in \sbra{l, r}} a\sbra{i}$.
In this paper, we propose a quantum data structure that supports RMQ queries and range updates, with an \textit{optimal} time complexity $\widetilde \Theta\rbra{\sqrt{nq}}$ for performing $q = O\rbra{n}$ operations \textit{without} preprocessing, compared to the classical $\widetilde\Theta\rbra{n+q}$.\footnote{$\widetilde \Theta\rbra{\cdot}$ suppresses logarithmic factors.}
As an application, we obtain a time-efficient quantum algorithm for $k$-minimum finding \textit{without} the use of \textit{quantum random access memory}. 
\end{abstract}


\newpage

\tableofcontents
\newpage

    \section{Introduction}

    Range Minimum Query (RMQ) is a basic problem, where we are tasked with finding (the index of) the minimum element in any subarray $a\substr{l}{r}$ of a given array $a\substr{1}{n}$ with $n$ elements from a totally ordered set. Formally, the task is to find
    \[
    \operatorname{RMQ}\rbra{a, l, r} = \argmin_{l \leq i \leq r} a\sbra{i}.
    \]
    In case of a tie, we define the answer as the smallest possible index.
    The RMQ problem, as a special case of orthogonal range queries \cite{Lue78,Wil85}, was initially studied in \cite{HT84} and \cite{GBT84} for solving graph and geometric problems, while also originated from another research line of parallel computing \cite{Val75,SV81,SV88,BV93}.
    Data structures for RMQ have a wide range of applications in text processing \cite{ALV92,Mut02,AKO04,FH06,Sad07,Sad07b,VM07,CPS08,FHS08,FMN08,HSV09,CIK+12} and graph problems \cite{RV88,GT04,BFCP+05,LC08}. 

    In this paper, we consider a general case of RMQ, where we are tasked with maintaining an online data structure that supports RMQ queries under updating the elements in the input array $a\substr{1}{n}$. 
    It is folklore that $q$ operations for RMQ can be performed in $O\rbra{n + q\log n}$ time, using textbook data structures for, e.g., range queries \cite{Lue78,BM80,Wil85} and dynamic trees \cite{ST83}.
    Recently, the time complexity was improved in \cite{NS14,BDR11} to $O\rbra{n + q \log n / \log \log n}$,\footnote{Here, the time complexity is measured in the word RAM model \cite{FW90}. The word RAM model is a common assumption for classical algorithms, where all (classical) data are stored in the memory as an array of $\Theta\rbra{\log n}$-bit words with $n$ the problem size. In our case, $n$ is the length of the input array $a\substr{1}{n}$.}
    which is optimal (as noted in \cite{Dav11}) because it was shown in \cite{AHR98} that if each update operation takes $\polylog\rbra{n}$ time, then each RMQ query requires $\Omega\rbra{\log n/\log \log n}$ time.\footnote{Let $t_{u}$ and $t_q$ be the time complexity (in the word RAM model) of each update operation and each RMQ query, respectively. Then, the worst-case time complexity of the RMQ data structure is $T = \Theta\rbra{n + q\rbra{t_u+t_q}}$ if its preprocessing time is $O\rbra{n}$. If $t_u = \Omega\rbra{\log n/\log \log n}$, then $T = \Omega\rbra{n + q \log n/\log \log n}$; if $t_u = O\rbra{\log n/\log \log n} = \polylog\rbra{n}$, then $t_q = \Omega\rbra{\log n/\log \log n}$ due to the lower bound given by \cite{AHR98}, and thus $T = \Omega\rbra{n + q\log n/\log \log n}$. Therefore, it always holds that $T = \Omega\rbra{n + q\log n/\log \log n}$, which concludes the optimality.}
    Further results include improving the space complexity \cite{HLMS16} and developing algorithms in the \textit{external memory} \cite{AV88} model \cite{AFSS13}. 

    \textbf{Motivation.}
    Quantum speedups are known for a wide range of computational problems, e.g., factoring \cite{Sho97}, search \cite{Gro96}, graph problems \cite{DHHM06}, and solving systems of linear equations \cite{HHL09}.
    These problems have been extensively studied in both classical and quantum settings.
    However, data structures, which serve as a cornerstone in designing efficient algorithms in the classical literature, remain not well-understood in quantum computing. 
    In this paper, we explore the quantum advantages in the RMQ problem and give a positive answer to the following question.
    \[
    \textit{Can quantum computing bring speedups to the RMQ problem?}
    \]

    \subsection{Main results}

    We study the quantum complexity of the RMQ problem. 
    We assume quantum query access to the input array $a\substr{1}{n}$, namely, a quantum unitary oracle $\mathcal{O}_a$
    that returns $a[i]$ given the index $i\in [n]$ in the following way:\footnote{We denote $\sbra{n} = \cbra{1, 2, \dots, n}$.}
    \[
    \mathcal{O}_a \ket{i, j, z} = \ket{i, j \oplus a\sbra{i}, z}
    \]
    for all $i,j,z$.
    The query complexity of a quantum algorithm for the RMQ problem is defined as the number of queries to the oracle $\mathcal{O}_a$, while the time complexity is defined as the sum of the query complexity, the number of elementary quantum gates, and the time complexity of classical operations. 
    In addition, we assume that the elements in $a\substr{1}{n}$ are from a totally ordered set $\rbra{A, \leq}$, with the operator $\min \colon A \times A \to A$ induced by ``$\leq$''. 
    For example, a common case of interest is $A = \mathbb{N} \cup \cbra{+\infty}$ with the natural definition of $\min\rbra{\cdot, \cdot}$.\footnote{Here, $+\infty$ is an element that satisfies $\min\rbra{x,+\infty}=\min\rbra{+\infty,x}=x$ for any $x\in \mathbb{N}$.}

    In this paper, we give tight (up to logarithmic factors) upper and lower bounds for the RMQ problem. 
    In the following, we use query complexity to measure the efficiency of our quantum algorithms, while they are also time-efficient with their time complexities only incurring an extra logarithmic factor. 
    It is worth noting that our quantum algorithm does not require Quantum Random Access Memory (QRAM) \cite{GLM08}, which is often needed to implement time-efficient quantum algorithms, e.g., \cite{vAGGdW20,AdW22}. 
    Due to the limitations of the current generation of quantum devices, it could be expensive to build a QRAM in practice regarding the required physical resources (see the discussions in, e.g., \cite{CHI+18}). 
    In contrast, we consider hybrid quantum-classical computation (cf. \cite{CC22}), and the time complexity of our quantum algorithms \textit{without} QRAM is measured by the sum of the query complexity, gate complexity (including quantum measurements), and the cost of classical computation in the word RAM model (see \cref{sec:def-quantum-algo} for the formal definition). 
    See implications in \cref{sec:imply} for the feature of not using QRAM.

    In the RMQ problem, it is asked to support both RMQ queries and range modifications in the online setting. 
    To formally define the operations for range modifications, let $\rbra{F \subseteq A^A, \circ}$ be a monoid with ``$\circ$'' the function composition and the identity element $\mathsf{id} \colon x \mapsto x$ for all $x \in A$, such that $f\rbra{\min\rbra{x, y}} = \min\rbra{f\rbra{x}, f\rbra{y}}$ for all $f \in F$ and $x, y \in A$. 
    The RMQ problem asks to maintain the following operations:
    \begin{itemize}
        \item $\mathsf{Initialize}\rbra{\mathcal{O}_a}$: Initialize with the array $a\substr{1}{n}$, with elements given by a quantum oracle $\mathcal{O}_a$.
        \item $\mathsf{Query}\rbra{l, r}$: Find the index of the minimum element among $a\sbra{l}, \dots, a\sbra{r}$, i.e., $\operatorname{RMQ}\rbra{a, l, r}$. 
        \item $\mathsf{Modify}\rbra{l, r, f}$: Set $a\sbra{i} \gets f\rbra{a\sbra{i}}$ for each $l \leq i \leq r$, where $f \in F$.
    \end{itemize}
    Here, $\mathsf{Initialize}\rbra{\mathcal{O}_a}$ is the preprocessing procedure which is called only once at the very beginning of the task. 
    We give several examples of $F$ that are of common interest.
    For the integer case with $A = \mathbb{N} \cup \cbra{+\infty}$, the set $F = \set{\mathsf{add}_c\rbra{x} = x+c}{c \in A}$ is valid, where $\mathsf{add}_c$ adds $c$ to the element; also, $F = \set{\mathsf{assign}_c\rbra{x} = c}{c \in A}$ is valid, where $\mathsf{assign}_c$ sets the element to $c$. 
    
    It is folklore that $\widetilde \Theta\rbra{n+q}$ time is necessary and sufficient to classically support $q$ operations for RMQ. 
    The classical lower bound follows from a simple observation that each element should be looked up at least once. 
    In sharp contrast, we show that this barrier, however, can be broken in the quantum world. 

    \begin{theorem} [Upper bound for RMQ, \cref{thm:QDynamicRMQ}] \label{thm:QDynamicRMQ-intro}
        There is a quantum algorithm for RMQ that supports $q \leq n$ operations on an array of length $n$ with query complexity $\widetilde O\rbra{\sqrt{nq}}$. 
    \end{theorem}

    Moreover, we show that \cref{thm:QDynamicRMQ-intro} is optimal up to polylogarithmic factors.

    \begin{theorem} [Lower bound for RMQ, \cref{thm:QDynamicRMQ-lower-bound}] \label{thm:QDynamicRMQ-lower-bound-intro}
        Any quantum algorithm for RMQ that supports $q$ operations on an array of length $n$ requires query complexity $\Omega\rbra{\sqrt{nq}}$. 
    \end{theorem}

    \cref{thm:QDynamicRMQ-intro} and \cref{thm:QDynamicRMQ-lower-bound-intro} together reveal an optimal quantum query complexity of $\widetilde \Theta\rbra{\sqrt{nq}}$ for RMQ. 
    We compare our quantum algorithm for RMQ with the classical counterparts in \cref{tab:dynamic}. 

\begin{table}[!htp]
\centering
\caption{Time complexity for RMQ.}
\normalsize
\label{tab:dynamic}
\begin{tabular}{ccc}
\toprule
Type & Time Complexity & References \\ \midrule
Classical & $O\rbra{n+q\log n}$ & \cite{ST83} \\
Classical & $\Theta\rbra{n+q\log n/\log \log n}$ & \cite{BDR11,AHR98} \\ \midrule
Quantum & $\widetilde \Theta\rbra{\sqrt{nq}}$ & Theorems \ref{thm:QDynamicRMQ-intro} and \ref{thm:QDynamicRMQ-lower-bound-intro} \\ \bottomrule
\end{tabular}
\end{table}

    \subsection{Implications} \label{sec:imply}
    \cref{thm:QDynamicRMQ-intro} implies a new quantum algorithm for $k$-minimum finding, namely, finding the smallest $k$ elements out of $n$. Quantum $k$-minimum finding serves as an important subroutine in many applications, including graph problems \cite{DHHM06}, quantum state preparation \cite{Ham22}, optimization \cite{LDR23,GJLW23}, etc. 
    Existing time-efficient implementations for quantum $k$-minimum finding require the assumption of QRAM. 
    As a corollary of \cref{thm:QDynamicRMQ-intro}, we present a time-efficient quantum algorithm for $k$-minimum finding \textit{without} QRAM. 

    \begin{corollary} [Quantum $k$-minimum finding without QRAM, \cref{corollary:qkmin}] \label{corollary:qkmin-intro}
        There is a quantum algorithm without QRAM for $k$-minimum finding with time complexity $O\rbra{\sqrt{nk\log\rbra{k \log n}}\log{n}}$.
    \end{corollary}

    As the Boolean case of $k$-minimum finding, finding $k$ marked elements (cf.\ \cite[Lemma 2]{dGdW02}) was recently known to have a time-efficient quantum implementation without QRAM in \cite{vAGN24}.
    In contrast to this work, their approach is specifically designed for the Boolean case, to which \cref{corollary:qkmin-intro} directly applies (though with a logarithmic overhead in the query complexity). 

    In comparison, the well-known quantum algorithm for $k$-minimum finding is due to \cite[Theorem 3.4]{DHHM06} with optimal query complexity $\Theta\rbra{\sqrt{nk}}$, which, however, is not time-efficient (as noted in \cite{vAGN24}) as only an implementation using $\widetilde O\rbra{\sqrt{n}k^{3/2}}$ elementary quantum gates is known if QRAM is not allowed.\footnote{If QRAM (in particular, quantum-read classical-write random access memory, a.k.a.\ QCRAM) is allowed, then the quantum $k$-minimum finding proposed in \cite{DHHM06} can be implemented using $O\rbra{\sqrt{nk} \log k}$ QRAM operations and with time complexity $O\rbra{\sqrt{nk} \log n}$.}
    It is noted that the Boolean result of \cite{vAGN24} implies a quantum algorithm for $k$-minimum finding with time complexity $O\rbra{\sqrt{nk} \rbra{\log^3k + \log n \log \log n} \log n}$.\footnote{This was noted by an anonymous reviewer. First, find the index $j$ of the, for example, $\rbra{2k+10}$-th smallest element with time complexity $O\rbra{\sqrt{nk}\log^2 n \log \log n}$ using the algorithm of \cite[Theorem 1.7]{NW99}. Then, the Boolean result of \cite[Theorem 1.1]{vAGN24} can be used to find all indices with value less than $a\sbra{j}$ with time complexity $O\rbra{\sqrt{nk}\log^3k \log n}$. Finally, (the indices of) the $k$ smallest elements can be found with classical postprocessing.
    The total time complexity of this approach is dominated by the sum of the time complexities of the first two steps, which is $O\rbra{\sqrt{nk} \rbra{\log^3k + \log n \log \log n} \log n}$.\label{footnote:another-qkmin}}
    Nevertheless, the time complexity of \cref{corollary:qkmin-intro} is better for all cases in polylogarithmic terms. 

    We compare the aforementioned results related to quantum $k$-minimum finding in \cref{tab:qkm}. 

\begin{table}[!htp]
\centering
\caption{Time complexity of quantum $k$-minimum finding.}
\normalsize
\label{tab:qkm}
\begin{tabular}{cccc}
\toprule
Case            & QRAM       & Time Complexity        & References                                     \\ \midrule
General & Y & $O\rbra{\sqrt{nk}\log n}$ & \cite{DHHM06} \\
General & N & $\widetilde O\rbra{\sqrt{n}k^{3/2}}$ & \cite{DHHM06} \\
General & N & $O\rbra{\sqrt{nk} \rbra{\log^3k + \log n \log \log n} \log n}$ & Implied by \cite{vAGN24,NW99} \\
General & N & $O\rbra{\sqrt{nk\log\rbra{k \log n}}\log{n}}$ & \cref{corollary:qkmin-intro} \\
\midrule
Boolean & Y & $O\rbra{\sqrt{nk}\log n}$ & \cite{dGdW02} \\
Boolean & N & $\widetilde O\rbra{\sqrt{n}k^{3/2}}$ & \cite{dGdW02} \\
Boolean & N & $O\rbra{\sqrt{nk}\log^3k \log n}$ & \cite{vAGN24} \\
\bottomrule
\end{tabular}
\end{table}

    \subsection{Techniques}

    \subsubsection{Upper bound}
    Our quantum algorithm in \cref{thm:QDynamicRMQ-intro} is under the framework of \textit{segment tree} (cf. \cite{BCKO08}, also known as the \textit{augmented B-tree} \cite{Pat08}). 
    A segment tree consists of $O\rbra{n}$ nodes to store an array of length $n$, with node $1$ the root for the range $\sbra{1, n}$. Node $k$ has left child node $2k$ and right child node $2k+1$ (if exists), and the range of node $k$ is the union of those of nodes $2k$ and $2k+1$. 
    For RMQ, each node maintains the minimum element over the range of that node. 
    Textbook solution to RMQ has time complexity $O\rbra{n+q\log{n}}$ for $q$ operations, and the current best solution has time complexity $O\rbra{n + q \log n / \log \log n}$ \cite{BDR11}. 
    
    Known classical approaches initialize the segment tree in $O\rbra{n}$ time. 
    In sharp contrast, we observe that the initialization is not necessary in the quantum case. 
    For this purpose, we introduce two new techniques to implementing segment trees that can take advantage of quantum computing: top-down completion and lazy node creation. 

\begin{figure}[!htp]
    \centering
    \includegraphics[scale=0.9]{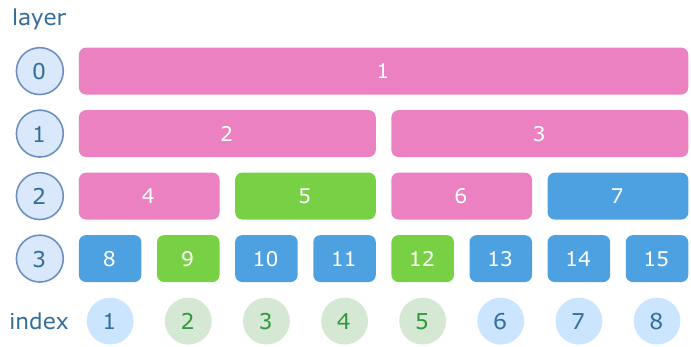}
    \caption{Segment tree for RMQ. In the example, $n = 8$ and the query range is $\sbra{2, 5}$ which is covered by (and thus involves) the green nodes $5, 9, 12$. In addition to the green nodes, the pink nodes $1, 2, 3, 4, 6$ are also visited, while the blue nodes are not visited.}
    \label{fig:dynamic}
\end{figure}

    \textit{Top-down completion}. In our quantum algorithm, we use the following key strategy: 
    \begin{itemize}
        \item Before the $k$-th operation, we create node $k$ (if it is valid and has not been created yet).
    \end{itemize}
    We will see the benefits of this strategy after introducing lazy node creation. At this stage, we note that the range of node $k$ has length $O\rbra{n/k}$, as node $k$ is in layer $O\rbra{\log k}$. 
    When node $k$ is created, we initialize it by quantum minimum finding in $\widetilde O\rbra{\sqrt{n/k}}$ time. 
    After $q$ operations, nodes numbered from $1$ to $q$ are created. Thus, the total cost of top-down completion is
    \[
    \sum_{k=1}^q \widetilde O\rbra*{\sqrt{\frac{n}{k}}} = \widetilde O\rbra{\sqrt{nq}}.
    \]

    \textit{Lazy node creation}. Except for the top-down completion strategy, we also need to visit the nodes that are not created yet when dealing with RMQ queries and range modifications. 
    For this, we use the following simple strategy:
    \begin{itemize}
        \item Create a node on the first visit. 
    \end{itemize}
    
    Now we will see the benefits of both strategies together. 
    Suppose that we are maintaining the $k$-th operation. 
    We note that an RMQ query or a range modification will involve at most $\ell = O\rbra{\log n}$ nodes in total, and at most $4$ nodes in each layer (see \cref{fig:dynamic} for an illustrative example). 
    Among these nodes, only those with number greater than $k$ (because of top-down completion) will possibly be created according to lazy node creation. 
    If we sort the nodes to be created by their numbers $m_1 < m_2 < \dots < m_\ell$ with $m_1 > k$, then it can be shown that $m_i \geq 2^{\floor{\rbra{i-1}/8}} m_1$ for every $1 \leq i \leq \ell$. 
    This means that the cost of lazy node creation in the $k$-th operation is 
    \[
    \sum_{i=1}^{\ell} \widetilde O\rbra*{\sqrt{\frac{n}{m_i}}} = \sum_{i=1}^{\ell} \widetilde O\rbra*{\sqrt{\frac{n}{2^{\floor{\rbra{i-1}/8}} k}}} = \widetilde O\rbra*{\sqrt{\frac{n}{k}}}.
    \]

    In summary, the total cost of both top-down completion and lazy node creation for maintaining $q$ operations is $\widetilde O\rbra{\sqrt{nq}}$, which gives the time complexity of our quantum algorithm for RMQ. 
    In detailed analysis for logarithmic terms, we use the small-error quantum minimum finding given in \cite{WY23}, and show that our quantum algorithm has query complexity $O\rbra{\sqrt{nq\log\rbra{q\log{n}}}}$, which is also time-efficient without QRAM (see \cref{thm:QDynamicRMQ} for details).
    
    \subsubsection{Lower bound}
    The lower bound in \cref{thm:QDynamicRMQ-lower-bound-intro} is based on a reduction from $k$-minimum finding to RMQ.
    The reduction is simple by repeating the following $k$ times: query the current global minimum element and delete it (in actual operation, set it to $+\infty$) using RMQ operations.
    Then, the lower bound for quantum RMQ comes from the lower bound $\Omega\rbra{\sqrt{nk}}$ for quantum $k$-minimum finding given in \cite{DHHM06} by setting $k = \Theta\rbra{q}$.

    \subsection{Related work} \label{sec:related-work}

    \textbf{Quantum data structures.}
    Quantum data structures were first studied in \cite{NABT15} known as quantum algorithms with (classical) advice, which were also further investigated in quantum cryptography known as the quantum random oracle model with auxiliary input (QROM-AI) \cite{HXY19}.
    There are also quantum data structures that support static Longest Common Extension (LCE) queries proposed in \cite{JN23,GT23}.
    Quantum $k$-minimum finding was recently generalized to the case with approximate values in \cite{GJW24}. 
    
    Other than quantum data structures, a special type of (classical) \textit{history-independent} data structures (cf.~\cite{YC22}) was found useful in designing time-efficient quantum algorithms for, e.g., element distinctness \cite{Amb07,BCJ+13}, triangle finding \cite{MSS07,JKM13,LG14}, subset sums \cite{BJLM13}, closest pair \cite{ACL+20}, and string problems \cite{LGS23,AJ23,JN23}. 

    \textbf{Static RMQ.}
    In the static scenario (where there is no update operation), the goal is to preprocess a fixed data structure $\mathfrak{D}$ for the given array $a\substr{1}{n}$, and then answer online RMQ queries with the help of $\mathfrak{D}$.\footnote{There are two types of schemes for static RMQ: \textit{systematic} and \textit{non-systematic} (also known as \textit{indexing} and \textit{encoding}, respectively). The systematic scheme allows access to both the input array $a$ and the additional preprocessed information $\mathfrak{D}$, while the non-systematic scheme only allows the additional preprocessed information $\mathfrak{D}$. For simplicity, we only consider the systematic scheme for static RMQ and we do not distinguish the two schemes.}
    The efficiency of $\mathfrak{D}$ is measured by a space-time pair $\ave{S, T}$, namely, $S$ is the size (in bits) of $\mathfrak{D}$ and $T$ is the time complexity of answering one RMQ query.
    A na{\"i}ve solution $\ave{\widetilde O\rbra{n^2}, O\rbra{1}}$ stores the answers of all $O\rbra{n^2}$ possible ranges and then just looks up the answer to each RMQ query; also, a brute force solution $\ave{O\rbra{1}, O\rbra{n}}$ preprocesses nothing but scans all the elements in the query range. 
    The textbook solution $\ave{O\rbra{n \log n}, O\rbra{1}}$ by \cite{HT84,GBT84} stores $O\rbra{n}$ \textit{words} and answers one RMQ query in constant time, via a reduction to the Lowest Common Ancestor (LCA) problem of the \textit{Cartesian tree} \cite{Vui80} of the array $a\substr{1}{n}$. 
    However, this solution is suboptimal in space (in bits). 
    The number of Cartesian trees of $n$ nodes is the $n$-th Catalan number $C_n = \binom{2n}{n}/\rbra{n+1}$, which implies an information-theoretic lower bound of $\log_2\rbra{C_n} = 2n - \Theta\rbra{\log n}$ bits for the storage of a Cartesian tree and thus for any data structure that supports constant-time RMQ queries. 
    The first linear-space constant-time solution $\ave{4n + O\rbra{n\log^2\log n/\log n}, O\rbra{1}}$ is due to \cite{Sad07b}. 
    It was later improved to $\ave{2n + O\rbra{n/\polylog\rbra{n}}, O\rbra{1}}$ in \cite{FH07,FH11,Pat08,NS14,DRS17}, with a space-time tradeoff $\ave{2n+n/\rbra{\log n/t}^{\Omega\rbra{t}}+\widetilde O\rbra{n^{3/4}}, O\rbra{t}}$ for any constant $t = O\rbra{1}$; 
    recently, an almost matching lower bound on the space complexity was shown in \cite{LY20,Liu21}.
    In addition to constant-time solutions, a space-time tradeoff $\ave{O\rbra{n/c}, O\rbra{c}}$ for any non-constant $c = O\rbra{n}$ was proposed in \cite{BDR12,FH11}, which is optimal due to the lower bound $ST = \Omega\rbra{n}$ given in \cite{BDR12}. 

    \textbf{Multidimensional RMQ.}
    RMQ is a special case of orthogonal range queries \cite{Lue78,Wil85}, with lower bounds and space-time tradeoffs studied in \cite{Fre81,Vai89}.
    It is also a special case of the \textit{partial sum} problem in the \textit{faithful} semigroup model \cite{Yao82}, with lower bounds and space-time tradeoffs studied in \cite{Yao82} and a matching upper bound given in \cite{CR89}. 

    The approach in \cite{GBT84} for one-dimensional RMQ implies an algorithm for $d$-dimensional (static) RMQ that supports RMQ queries in time $O\rbra{\log^{d-1}{n}}$. 
    Later, an improved algorithm was proposed in \cite{CR89} that achieves $O\rbra{\rbra{\alpha\rbra{n}}^d}$ time per RMQ query after $O\rbra{n}$ preprocessing time, where $\alpha\rbra{n}$ \cite{Tar75} is the functional inverse of the Ackermann function \cite{Ack28,Pet35,Rob48}.
    For $2$-dimensional RMQ, this was later improved to $O\rbra{1}$ time per query and $O\rbra{n \log^{o\rbra{1}}{n}}$ preprocessing time \cite{AFL07}. 
    Finally, an algorithm for $d$-dimensional RMQ for any constant $d$ was proposed in \cite{YA10} that supports RMQ query in $O\rbra{1}$ time, after $O\rbra{n}$ preprocessing. 
    Recently, a space-time tradeoff
    $\ave{O\rbra{n/c}, O\rbra{c \log^2{c}}}$ for $2$-dimensional RMQ was proposed in \cite{BDR12}.

    Another line of work considers the \textit{discrete} (static) RMQ problem where $n$ points on a $2$-dimensional plane are given with priority values. 
    This problem was initiated by \cite{Wil86,Cha88}.
    In \cite{Cha88}, an algorithm was proposed that uses $O\rbra{n}$ words and supports RMQ queries in $O\rbra{\log^{1+\varepsilon}{n}}$ time, which was recently improved to $O\rbra{\log{n}\log\log{n}}$ time per query with the same space complexity \cite{FMR12}. 
    Better time complexity for queries was recently achieved at a cost of increasing space. 
    In \cite{CLP11}, an algorithm was proposed that uses $O\rbra{n\log^{\varepsilon}{n}}$ words and supports RMQ queries in $O\rbra{\log \log {n}}$ time; another algorithm can be obtained by combining the results of \cite{KN09} and \cite{Cha13} (as noted in \cite{Nek21}). 
    Recently, two better algorithms were developed in \cite{Nek21}: one uses $O\rbra{n}$ words that supports RMQ queries in $O\rbra{\log^{\varepsilon}{n}}$ time, the other uses $O\rbra{n\log\log{n}}$ words that supports RMQ queries in $O\rbra{\log\log{n}}$ time.

    \subsection{Discussion}

    In this paper, we give tight (up to logarithmic factors) quantum upper and lower bounds for the RMQ problem. 
    A notable strength of our quantum algorithms is that they have time-efficient implementations without the assumption of QRAM. 
    As an application, we obtain a time-efficient quantum $k$-minimum finding without QRAM. 
    We mention some possible directions that are related to this work. 

    \begin{enumerate}
        \item Is it possible to reduce/remove the logarithmic factors in the query/time complexity of our quantum algorithms? 
        A possible improvement may be achieved by adopting the modern data structures in \cite{NS14,BDR11} into quantum computing. \label{ques:1}
        \item As mentioned in \cref{sec:related-work}, multidimensional RMQs have been extensively studied in the classical literature, with practical applications in computational geometry. 
        An interesting direction is to study how quantum computing can speedup multidimensional RMQs. 
        \item Can we find other data structure problems that can take advantage of quantum computing? 
        \item In this work, we find a novel way to reduce the requirement of QRAM. 
        We hope this could inspire other quantum algorithms without QRAM. 
    \end{enumerate}

    \section{Preliminaries}

    In this section, we introduce the basic concepts and tools for quantum query algorithms, and review textbook classical data structures for RMQ. 

    \subsection{Quantum query algorithms} \label{sec:def-quantum-algo}

    A quantum query algorithm $\mathcal{A}$ using queries to a quantum unitary oracle $\mathcal{O}$ can be described by a quantum circuit 
    \[
    \mathcal{A} = U_T \mathcal{O} \dots \mathcal{O} U_1 \mathcal{O} U_0,
    \]
    where $U_t$ are quantum unitary gates independent of $\mathcal{O}$ for $0 \leq t \leq T$. 
    The query complexity of $\mathcal{A}$ is defined as $\mathsf{Q}\rbra{\mathcal{A}} = T$. 
    The gate complexity of $\mathcal{A}$ is defined as $\mathsf{G}\rbra{\mathcal{A}} = \sum_{t=0}^T \mathsf{G}\rbra{U_t}$, where $\mathsf{G}\rbra{U_t}$ denotes the number of one- and two-qubit quantum gates that are sufficient to implement the unitary operator $U_t$.
    Suppose that $\mathcal{A}$ consists of two subsystems $\mathsf{W}$ (work) and $\mathsf{O}$ (output). 
    Then, the output of $\mathcal{A}$ is the outcome of the subsystem $\mathsf{O}$ after measuring $\mathcal{A}\ket{0}_{\mathsf{W}}\ket{0}_{\mathsf{O}}$ in the computational basis. 
    In other words, $\mathcal{A}$ will output $y$ with probability $\Abs{\bra{y}_{\mathsf{O}}\mathcal{A}\ket{0}_{\mathsf{W}}\ket{0}_{\mathsf{O}}}^2$.

    To analyze the time efficiency in this paper, we consider the hybrid quantum-classical computation (cf. \cite{CC22}) as follows. 
    A hybrid quantum-classical algorithm can be described by a sequence of procedures 
    \[
    \mathcal{P} \colon \mathcal{B}_0 \to \mathcal{A}_1 \to \mathcal{B}_1 \to \mathcal{A}_2 \to \mathcal{B}_2 \to \dots \to \mathcal{A}_m \to \mathcal{B}_m
    \]
    for some $m \geq 0$, 
    where $\mathcal{A}_i$ is a quantum query algorithm using queries to $\mathcal{O}$ for $1 \leq i \leq m$, and $\mathcal{B}_i$ is a classical algorithm independent of $\mathcal{O}$. 
    The hybrid algorithm $\mathcal{P}$ runs as follows:
    \begin{enumerate}
        \item For each $0 \leq i \leq m$, $\mathcal{B}_i$ has (word RAM) access to the (classical) outputs of the previous procedures $\mathcal{B}_0, \mathcal{A}_1, \mathcal{B}_1, \dots, \mathcal{A}_i$.
        \item For each $0 \leq i < m$, $\mathcal{B}_i$ outputs the classical description of $\mathcal{A}_{i+1}$.
        \item The output of $\mathcal{P}$ is defined as the output of $\mathcal{B}_m$.
    \end{enumerate}
    The query complexity of $\mathcal{P}$ is defined as $\mathsf{Q}\rbra{\mathcal{P}} = \sum_{i=1}^m \mathsf{Q}\rbra{\mathcal{A}_i}$. 
    The time complexity of $\mathcal{P}$ is defined as $\mathsf{T}\rbra{\mathcal{P}} = \mathsf{Q}\rbra{\mathcal{P}} + \sum_{i=1}^m \mathsf{G}\rbra{\mathcal{A}_i} + \sum_{i=0}^m \mathsf{T}\rbra{\mathcal{B}_i}$, where $\mathsf{T}\rbra{\mathcal{B}_i}$ is the (word RAM) time complexity of $\mathcal{B}_i$. 
    Roughly speaking, the time complexity of a hybrid quantum-classical algorithm is the sum of the query complexity, the gate complexity, and the word RAM time cost. 
    
    In comparison, the time complexity of quantum query algorithms with QRAM (quantum-read classical-write random access memory) defined in \cite{AdW22} is the sum of the query complexity, the gate complexity, and the QRAM time cost. 
    Here, the QRAM time cost is the number of QRAM operations performed during the execution of the quantum algorithm, where each QRAM operation can modify or retrieve information from an array $x\substr{1}{M}$, with $M$ the length and each $x\sbra{i}$ initialized to $0$, in the following manner:
    \begin{itemize}
        \item Modification: Given an index $i$ and a bit $b$, set $x\sbra{i}$ to $b$. 
        \item QRAM query: Perform the unitary operator $U_{\textup{QRAM}} \colon \ket{i}\ket{j} \mapsto \ket{i}\ket{j\oplus x\sbra{i}}$ on any set of qubits. 
    \end{itemize}
    Roughly speaking, QRAM operations allow us to read classical data coherently on a quantum computer, whereas word RAM operations only allow us to read classical data via classical addressing.

    \textit{Throughout this paper, the time-efficiency of all quantum algorithms are measured under the above framework of hybrid quantum-classical computation, unless otherwise stated.}

    \subsection{Quantum minimum finding}

    Quantum minimum finding was first studied in \cite{DH96} with Grover's algorithm \cite{Gro96} as a subroutine.
    Later in \cite{DHHM06}, it was generalized to finding the smallest $k$ elements and was used as a subroutine in solving graph problems. 
    To deal with small errors, we need the following version of quantum minimum finding given in \cite{WY23}. 
    \begin{lemma} [{\cite[Lemma 3.5]{WY23}}]
        \label{lemma:findmin}
        Given quantum oracle $\mathcal{O}_a$ to array $a\substr{1}{n}$, there is a quantum query algorithm $\mathsf{findmin}\rbra{a\substr{1}{n}, \varepsilon}$ that outputs $\argmin_{i \in \sbra{n}} a\sbra{i}$ with probability at least $1-\varepsilon$ with query complexity $O\rbra{\sqrt{n\log\rbra{1/\varepsilon}}}$ and time complexity $O\rbra{\sqrt{n\log\rbra{1/\varepsilon}} \log{n}}$.
    \end{lemma}

    \section{Quantum Data Structure for RMQ}

    In this section, we investigate the quantum advantages in RMQ. 
    For completeness, we review a textbook classical approach for RMQ in \cref{sec:classic-dynamic-rmq}. 
    In \cref{sec:upper-bound-dynamic-rmq}, we provide a quantum algorithm for RMQ based on the framework of \cref{sec:classic-dynamic-rmq}. 
    In \cref{sec:k-min-find}, we apply our quantum RMQ to an implementation of quantum $k$-minimum finding without QRAM.
    Finally, we give a matching lower bound on the quantum query complexity of RMQ in \cref{sec:lower-bound-dynamic-rmq}.

    \subsection{A Classical Approach for RMQ} \label{sec:classic-dynamic-rmq}

We first recall that RMQ requires to support both RMQ queries and range modifications on the input array $a\substr{1}{n}$:
\begin{itemize}
    \item $\mathsf{Query}\rbra{l, r}$: Return $\operatorname{RMQ}\rbra{a, l, r}$. 
    \item $\mathsf{Modify}\rbra{l,r,f}$: Set $a\sbra{i} \gets f\rbra{a\sbra{i}}$ for each $l \leq i \leq r$, where $f$ is a function that satisfies $f\rbra{\min\rbra{x, y}} = \min\rbra{f\rbra{x}, f\rbra{y}}$.\footnote{Let $A$ be the set of possible values of the elements in the input array $a\substr{1}{n}$, and let $F$ be the set of all functions $f$ allowed in range modifications. We assume that $F \subseteq A^A$ is a monoid with ``$\circ$'' the function composition and the identity element $\mathsf{id}\colon x \mapsto x$, such that $f\rbra{\min\rbra{x, y}} = \min\rbra{f\rbra{x}, f\rbra{y}}$ for all $f \in F$. For the time efficiency, we assume that: (i) any $f \in F$ and $x \in A$ can be stored in $O\rbra{1}$ words; (ii) the composition $f \circ g$ can be computed in $O\rbra{1}$ time for any $f, g \in F$; and (iii) the evaluation $f\rbra{x}$ can be computed in $O\rbra{1}$ time for any $f \in F$ and $x \in A$.}
\end{itemize}

A textbook (classical) solution to RMQ is based on segment trees. 
A segment tree $\mathcal{T}$ for an array $a\substr{1}{n}$ consists of $O\rbra{n}$ nodes. 
Every node of the segment tree $\mathcal{T}$ is associated with a number $k$, denoted $\mathsf{nodes}\sbra{k}$, and it has four basic attributes: $l, r, \mathsf{lchild}, \mathsf{rchild}$, where $\sbra{l, r}$ is the index range it represents and $\mathsf{lchild}$ and $\mathsf{rchild}$ are its left and right children, respectively.
These basic attributes of each node are defined recursively as follows:
\begin{enumerate}
    \item $\mathsf{nodes}\sbra{1}.l = 1$ and $\mathsf{nodes}\sbra{1}.r = n$. 
    \item For every $k \geq 1$ with $\mathsf{nodes}\sbra{k}.l$ and $\mathsf{nodes}\sbra{k}.r$ already defined, if $\mathsf{nodes}\sbra{k}.l < \mathsf{nodes}\sbra{k}.r$, then define 
    \begin{itemize}
        \item $\mathsf{nodes}\sbra{k}.\mathsf{lchild} = \mathsf{nodes}\sbra{2k}$, $\mathsf{nodes}\sbra{2k}.l = \mathsf{nodes}\sbra{k}.l$, $\mathsf{nodes}\sbra{2k}.r = m$,
        \item $\mathsf{nodes}\sbra{k}.\mathsf{rchild} = \mathsf{nodes}\sbra{2k+1}$, $\mathsf{nodes}\sbra{2k+1}.l = m+1$, $\mathsf{nodes}\sbra{2k+1}.r = \mathsf{nodes}\sbra{k}.r$, 
    \end{itemize}
    where $m = \floor{\rbra{\mathsf{nodes}\sbra{k}.l+\mathsf{nodes}\sbra{k}.r}/2}$.
\end{enumerate}
Without loss of generality, we may assume that $n = 2^s$ for an integer $s \geq 1$; then, $\mathsf{nodes}\sbra{k}.l$ and $\mathsf{nodes}\sbra{k}.r$ are defined for every $1 \leq k \leq 2n-1$. 

To solve RMQ with a segment tree, we assign two extra attributes $v, g$ to each node, where $v$ means the range minimum and $g$ is the lazy tag. 
An implementation of the initialization procedure is given in \cref{algo:Initialize}, where each node $\mathsf{nodes}\sbra{k}$ is initialized with $v$ set to the range minimum and $g$ the identity function $\mathsf{id} \colon x \mapsto x$. 
The attribute $v$ of each node is computed through calling $\mathsf{update}$ after the initialization of its left and right children is done. 
It can be verified that \cref{algo:Initialize} has time complexity $O\rbra{n}$. 

\begin{algorithm}[!htp]
    \caption{$\mathsf{Initialize}\rbra{a\substr{1}{n}}$}
    \label{algo:Initialize}
    \begin{algorithmic}[1]
    \Procedure{$\mathsf{update}$}{$\mathsf{node}$}
        \State $\mathsf{node}.v \gets \min\rbra{\mathsf{node}.\mathsf{lchild}.v, \mathsf{node}.\mathsf{rchild}.v}$;
    \EndProcedure
    \Procedure{$\mathsf{initialize}$}{$k, l, r$}
    \State $\mathsf{nodes}\sbra{k}.l \gets l$;
    \State $\mathsf{nodes}\sbra{k}.r \gets r$;
    \State $\mathsf{nodes}\sbra{k}.g \gets \mathsf{id}$;
    \If {$l = r$}
        \State $\mathsf{nodes}\sbra{k}.v \gets a\sbra{l}$;
    \Else
        \State $m \gets \floor{\rbra{l+r}/2}$;
        \State $\mathsf{initialize}\rbra{2k, l, m}$;
        \State $\mathsf{initialize}\rbra{2k+1, m+1, r}$;
        \State $\mathsf{nodes}\sbra{k}.\mathsf{lchild} \gets \mathsf{nodes}\sbra{2k}$;
        \State $\mathsf{nodes}\sbra{k}.\mathsf{rchild} \gets \mathsf{nodes}\sbra{2k+1}$;
        \State $\mathsf{update}\rbra{\mathsf{nodes}\sbra{k}}$;
    \EndIf
    
    \EndProcedure
    \State \Return $\mathsf{initialize}\rbra{1, 1, n}$;
    \end{algorithmic}
\end{algorithm}

\textit{Range modifications}. The implementation of the range modifications is given in \cref{algo:Modify}. The idea is to search from the root $\mathsf{nodes}\sbra{1}$ of the segment tree and modify the lazy tags of at most $O\rbra{\log n}$ nodes. 
Specifically, considering a call to $\mathsf{modify}(\mathsf{node},l,r,f)$, we consider the two cases provided that  $[\mathsf{nodes}\sbra{k}.l,\mathsf{nodes}\sbra{k}.r]\cap[l,r]\neq\emptyset$: 
\begin{itemize}
    \item Type I: $[\mathsf{nodes}\sbra{k}.l,\mathsf{nodes}\sbra{k}.r]\not\subseteq[l,r]$. 
    Note that in each layer, there are at most two nodes $\mathsf{nodes}\sbra{k}$ satisfying $[\mathsf{nodes}\sbra{k}.l,\mathsf{nodes}\sbra{k}.r]\cap[l,r]\not=\emptyset$ and  $[\mathsf{nodes}\sbra{k}.l,\mathsf{nodes}\sbra{k}.r]\subseteq[l,r]$. 
    So we can recalculate these type I nodes from the information of their children in $O(\log n)$ time. 
    However, the number of nodes satisfying $[\mathsf{nodes}\sbra{k}.l,\mathsf{nodes}\sbra{k}.r]\subseteq[l,r]$ is much larger, and thus we cannot recalculate them individually. 
    To resolve this issue, we maintain an operator $g$ (which is an identity operator initially) in every node, where the operator $g$ of $\mathsf{nodes}\sbra{k}$ means that we would apply $g$ to all descendants of $\mathsf{nodes}\sbra{k}$. 
    \item Type II: $[\mathsf{nodes}\sbra{k}.l,\mathsf{nodes}\sbra{k}.r]\subseteq[l,r]$. 
    For every type II node $\mathsf{nodes}\sbra{k}$, we directly apply $f$ to $\mathsf{nodes}\sbra{k}.v$ and merge $f$ into $\mathsf{nodes}\sbra{k}.g$. These can be implemented efficiently under the assumption that the multiplication and evalutaion of operators $g$ takes $O(1)$ time. 
\end{itemize}
To ensure the correctness, every time we will visit the children of $\mathsf{node}$, we call the $\mathsf{pushdown}\rbra{\mathsf{node}}$ procedure to propagate the lazy tag $g$ of $\mathsf{node}$ to its children. 

\begin{algorithm}[t]
    \caption{$\mathsf{Modify}\rbra{l, r, f}$}
    \label{algo:Modify}
    \begin{algorithmic}[1]
    \Procedure{$\mathsf{pushdown}$}{$\mathsf{node}$}
        \State $\mathsf{node}.\mathsf{lchild}.v \gets \mathsf{node}.g \rbra{\mathsf{node}.\mathsf{lchild}.v}$;
        \State $\mathsf{node}.\mathsf{rchild}.v \gets \mathsf{node}.g \rbra{\mathsf{node}.\mathsf{rchild}.v}$;
        \State $\mathsf{node}.\mathsf{lchild}.g \gets \mathsf{node}.g \circ \mathsf{node}.\mathsf{lchild}.g$;
        \State $\mathsf{node}.\mathsf{rchild}.g \gets \mathsf{node}.g \circ \mathsf{node}.\mathsf{rchild}.g$;
        \State $\mathsf{node}.g \gets \mathsf{id}$;
    \EndProcedure
    \Procedure{$\mathsf{modify}$}{$\mathsf{node}, l, r, f$}
    \If {$\sbra{\mathsf{node}.l, \mathsf{node}.r} \subseteq \sbra{l, r}$}
        \State $\mathsf{node}.v \gets f\rbra{\mathsf{node}.v}$;
        \State $\mathsf{node}.g \gets f \circ \mathsf{node}.g$;
        \State \Return;
    \EndIf
    \State $\mathsf{pushdown}\rbra{\mathsf{node}}$;
    \If {$\sbra{\mathsf{node}.\mathsf{lchild}.l, \mathsf{node}.\mathsf{lchild}.r} \cap \sbra{l, r} \neq \emptyset$}
        \State $\mathsf{modify}\rbra{\mathsf{node}.\mathsf{lchild}, l, r, f}$;
    \EndIf
    \If {$\sbra{\mathsf{node}.\mathsf{rchild}.l, \mathsf{node}.\mathsf{rchild}.r} \cap \sbra{l, r} \neq \emptyset$}
        \State $\mathsf{modify}\rbra{\mathsf{node}.\mathsf{rchild}, l, r, f}$;
    \EndIf
    \State $\mathsf{update}\rbra{\mathsf{node}}$;
    \EndProcedure
    \State $\mathsf{modify}\rbra{\mathsf{nodes}\sbra{1}, l, r, f}$;
    \end{algorithmic}
\end{algorithm}

\textit{RMQ Queries}. In order to answer RMQ queries, we search from the root $\mathsf{nodes}\sbra{1}$ of the segment tree and visit at most $O\rbra{\log n}$ nodes that are sufficient to compute the RMQ. 
An implementation for answering RMQ queries is given in \cref{algo:Query}.
In particular, each time we visit a node during this procedure, we need to propagate the lazy tag $g$ of that node to its children. 

\begin{algorithm}[t]
    \caption{$\mathsf{Query}\rbra{l, r}$}
    \label{algo:Query}
    \begin{algorithmic}[1]
    \Function{$\mathsf{query}$}{$\mathsf{node}, l, r$}
    \If {$\sbra{\mathsf{node}.l, \mathsf{node}.r} \subseteq \sbra{l, r}$}
        \State \Return $\mathsf{node}.v$;
    \EndIf
    \State $\mathsf{pushdown}\rbra{\mathsf{node}}$;
    \If {$\sbra{\mathsf{node}.\mathsf{rchild}.l, \mathsf{node}.\mathsf{rchild}.r} \cap \sbra{l, r} = \emptyset$}
        \State \Return $\mathsf{query}\rbra{\mathsf{node}.\mathsf{lchild}, l, r}$;
    \ElsIf {$\sbra{\mathsf{node}.\mathsf{lchild}.l, \mathsf{node}.\mathsf{lchild}.r} \cap \sbra{l, r} = \emptyset$}
        \State \Return $\mathsf{query}\rbra{\mathsf{node}.\mathsf{rchild}, l, r}$;
    \Else
        \State \Return $\min\rbra{ \mathsf{query}\rbra{\mathsf{node}.\mathsf{lchild}, l, r}, \mathsf{query}\rbra{\mathsf{node}.\mathsf{rchild}, l, r} }$;
    \EndIf
    \EndFunction
    \State \Return $\mathsf{query}\rbra{\mathsf{nodes}\sbra{1}, l, r}$;
    \end{algorithmic}
\end{algorithm}

We can implement a data structure for RMQ with
Algorithms \ref{algo:Initialize}, \ref{algo:Modify}, and \ref{algo:Query}. 
We omit the simple proof of its correctness and complexity. 
\begin{lemma}[Classical RMQ]
    The RMQ implemented by
    Algorithms \ref{algo:Initialize}, \ref{algo:Modify}, and \ref{algo:Query} has time complexity $O\rbra{n + q\log n}$ for $q$ operations on an array of length $n$.
\end{lemma}
It is worth noting that in \cite{BDR11}, the time complexity for RMQ was recently improved to $O\rbra{n + q \log n / \log \log n}$.

    \subsection{Quantum upper bound} \label{sec:upper-bound-dynamic-rmq}

    We provide a quantum algorithm for RMQ as follows. 

    \begin{theorem} [Quantum RMQ] \label{thm:QDynamicRMQ}
        There is a quantum algorithm for RMQ that supports $q = O\rbra{n}$ operations on an array of length $n$ with overall success probability at least $1-\varepsilon$ with query complexity $O\rbra{\sqrt{nq\log\rbra{q\log\rbra{n}/\varepsilon}}}$ and time complexity $O\rbra{\sqrt{nq\log\rbra{q\log\rbra{n}/\varepsilon}}\log{n}}$. 
        More precisely, there is a quantum data structure $\mathsf{QDynamicRMQ}$ that supports the following operations:
        \begin{itemize}
            \item Initialization: $\mathsf{QDynamicRMQ}.\mathsf{Initialize}\rbra{\mathcal{O}_a, q, \varepsilon}$ returns an instance of $\mathsf{QDynamicRMQ}$ initialized with array $a\substr{1}{n}$, the number $q$ of operations, and the required error parameter $\varepsilon$, with $O\rbra{1}$-time preprocessing and no queries to $\mathcal{O}_a$. 
            \item Query: $\mathsf{QDynamicRMQ}.\mathsf{Query}\rbra{l, r}$ returns $\operatorname{RMQ}\rbra{a, l, r}$. 
            \item Modification: $\mathsf{QDynamicRMQ}.\mathsf{Modify}\rbra{l, r, f}$ sets $a\sbra{i} \gets f\rbra{a\sbra{i}}$ for every $l \leq i \leq r$, provided that $f \in F$. 
        \end{itemize}
        The data structure $\mathsf{QDynamicRMQ}$ can maintain $q$ query/modification operations on $a\substr{1}{n}$ with overall success probability at least $1-\varepsilon$ using $O\rbra{\sqrt{nq\log\rbra{q\log\rbra{n}/\varepsilon}}}$ queries to $\mathcal{O}_a$.
        Moreover, if the exact value of $q \leq n$ is not known in advance, then the query complexity is $O\rbra{\sqrt{nq\log\rbra{n/\varepsilon}}}$ and the time complexity is $O\rbra{\sqrt{nq\log\rbra{n/\varepsilon}}\log n}$.
    \end{theorem}

    \begin{algorithm}[t]
        \caption{$\mathsf{QDynamicRMQ}$}
        \label{algo:qdynamic-rmq}
        \begin{algorithmic}[1]
        \Procedure{$\mathsf{create}$}{$\mathsf{node}$}
        \If {$\mathsf{create}\rbra{\mathsf{node}}$ has not ever been called}
            \State Let $l \gets \mathsf{node}.l$ and $r \gets \mathsf{node}.r$ be the index range associated with $\mathsf{node}$. 
            \State $\mathsf{node}.v \gets \mathsf{findmin}\rbra{a\substr{l}{r}, \varepsilon/\rbra{4q\log_2 n}}$;
            \State $\mathsf{node}.g \gets \mathsf{id}$;
        \EndIf
        \EndProcedure
        \Procedure{$\mathsf{pushdown}$}{$\mathsf{node}$}
            \State $\mathsf{create}\rbra{\mathsf{node}.\mathsf{lchild}}$; 
            \State $\mathsf{create}\rbra{\mathsf{node}.\mathsf{rchild}}$;
            \State $\mathsf{node}.\mathsf{lchild}.v \gets \mathsf{node}.g \rbra{\mathsf{node}.\mathsf{lchild}.v}$;
            \State $\mathsf{node}.\mathsf{rchild}.v \gets \mathsf{node}.g \rbra{\mathsf{node}.\mathsf{rchild}.v}$;
            \State $\mathsf{node}.\mathsf{lchild}.g \gets \mathsf{node}.g \circ \mathsf{node}.\mathsf{lchild}.g$;
            \State $\mathsf{node}.\mathsf{rchild}.g \gets \mathsf{node}.g \circ \mathsf{node}.\mathsf{rchild}.g$;
            \State $\mathsf{node}.g \gets \mathsf{id}$;
        \EndProcedure
        \Procedure{$\mathsf{Initialize}$}{$\mathcal{O}_a, q, \varepsilon$}
            \State $k \gets 0$;
        \EndProcedure
        \Function{$\mathsf{Query}$}{$l, r$}
            \State $k \gets k + 1$;
            \State $\mathsf{create}\rbra{\mathsf{nodes}\sbra{k}}$;
            \State \Return $\mathsf{query}\rbra{\mathsf{nodes}\sbra{1}, l, r}$;
        \EndFunction
        \Procedure{$\mathsf{Modify}$}{$l, r, f$}
            \State $k \gets k + 1$;
            \State $\mathsf{create}\rbra{\mathsf{nodes}\sbra{k}}$;
            \State $\mathsf{modify}\rbra{\mathsf{nodes}\sbra{1}, l, r, f}$;
        \EndProcedure
        \end{algorithmic}
    \end{algorithm}

    \begin{proof}
        Without loss of generality, we assume that $n=2^s$ for some positive integer $s$. The segment tree on the input array $a\substr{1}{n}$ can be divided into $s+1$ layers, labeled $0,\ldots,s$ from top to bottom, where the layer $i$ has exactly $2^i$ nodes. For layer $i$, $\mathsf{nodes}\sbra{2^i},\ldots,\mathsf{nodes}\sbra{2^{i+1}-1}$ is placed from left to right. 
        Under this framework, $\mathsf{nodes}\sbra{k}$ has left child $\mathsf{nodes}\sbra{2k}$ and right child $\mathsf{nodes}\sbra{2k+1}$ for $1\le k<n$, and $\mathsf{nodes}\sbra{k}$ belongs to layer $\lfloor\log_2 k\rfloor$ for $1\le k<2n$. 
        See \cref{fig:dynamic} for an illustrative example. 

        Next, we introduce two new ideas from a high level: top-down completion and lazy node creation. 
        
        \textit{Top-down completion}. Instead of preprocessing nodes in advance, we can create $\mathsf{nodes}\sbra{k}$ and assign the answer of the range it represents, via quantum minimum finding, before the $k$-th operation. In the stage of initialization, $k$, the number of operations that have been processed till now, is set to $0$. As we will see, the top-down completion guarantees the query complexity and the time complexity of the $k$-th operation in the worst case.

        \textit{Lazy node creation}. 
        In sharp contrast to the ordinary segment tree where the $\mathsf{pushdown}$ procedure (in \cref{algo:Modify}) only updates necessary information, we design a quantum $\mathsf{pushdown}$ procedure where the left child and the right child of a node will be created if either of them has not been created yet. This ensures that each node has already created and initialized before visiting it.

        Based on top-down completion and lazy node creation, we override $\mathsf{Initialize}$, $\mathsf{Query}$, $\mathsf{Modify}$, and $\mathsf{pushdown}$ in Algorithms \ref{algo:Initialize}, \ref{algo:Modify}, and \ref{algo:Query} by the implementations in \cref{algo:qdynamic-rmq}.\footnote{Here, \cref{algo:qdynamic-rmq} implements the quantum RMQ, where a new function $\mathsf{create}$ is introduced, the functions $\mathsf{Initialize}$, $\mathsf{Query}$, $\mathsf{Modify}$, and $\mathsf{pushdown}$ have new implementations, and the remaining functions have the same implementations as in Algorithms \ref{algo:Initialize}, \ref{algo:Modify}, and \ref{algo:Query}.
        It is worth noting that the implementation of quantum RMQ in \cref{algo:qdynamic-rmq} uses the same functions $\mathsf{modify}$ and $\mathsf{query}$ with the same implementations as in \cref{algo:Modify,algo:Query}, respectively, whereas the function $\mathsf{pushdown}$ called by them uses the new implementation given in \cref{algo:qdynamic-rmq}.}
        We explain them in detail as follows. 
        \begin{itemize}
            \item $\mathsf{Initialize}\rbra{\mathcal{O}_a, q, \varepsilon}$. Rather than initialize the whole segment tree, we do nothing except for setting $k$, the number of operations that have been processed till now, to $0$. (It should be noted that we also store the values of $q$ and $\varepsilon$ for future use.)
            \item $\mathsf{Query}\rbra{l, r}$. Compared to the classical implementation given in \cref{algo:Query}, before calling $\mathsf{query}\rbra{\mathsf{nodes}\sbra{1}, l, r}$, we increment $k$ by $1$ and then create node $k$ by calling $\mathsf{create}\rbra{\mathsf{nodes}\sbra{k}}$ in \cref{algo:qdynamic-rmq}.
            \item $\mathsf{Modify}\rbra{l, r, f}$. Similar to $\mathsf{Query}\rbra{l, r}$, before calling $\mathsf{modify}\rbra{\mathsf{nodes}\sbra{1}, l, r, f}$, we increment $k$ by $1$ and then create node $k$ by calling $\mathsf{create}\rbra{\mathsf{nodes}\sbra{k}}$ in \cref{algo:qdynamic-rmq}.
            \item $\mathsf{pushdown}\rbra{\mathsf{node}}$.
            Compared to the classical implementation given in \cref{algo:Modify}, before updating necessary information, we create the left and right children of $\mathsf{node}$ by calling $\mathsf{create}\rbra{\mathsf{node}.\mathsf{lchild}}$ and $\mathsf{create}\rbra{\mathsf{node}.\mathsf{rchild}}$. 
        \end{itemize}
        The above procedures heavily depend on the new added special procedure $\mathsf{create}\rbra{\mathsf{node}}$.
        \begin{itemize}
            \item $\mathsf{create}\rbra{\mathsf{node}}$. 
            Initiate all the attributes of $\mathsf{node}$. Especially, $\mathsf{node}.v$ stores (the index of) the minimum in the range represented by $\mathsf{node}$; this is done by the quantum minimum finding in \cref{lemma:findmin} with success probability at least $1 - \varepsilon/\rbra{4q\log_2 n}$.
        \end{itemize}

        \textit{Correctness}. 
        It can be seen that each RMQ query or range modification will call the $\mathsf{create}$ procedure at most $4s \leq 4\log_2 n$ time. 
        This is because any range $\sbra{l, r}$ will only involve at most $4$ nodes in each layer except for layer $0$ (with $1$ node) and layer $1$ (with $2$ nodes). 
        As a result, the overall success probability of the quantum minimum finding is 
        \[
        \rbra*{1 - \frac{\varepsilon}{4q\log_2 n}}^{4s} \geq 1 - \frac{4s\varepsilon}{4q\log_2 n} \geq 1 - \varepsilon.
        \]
        Except for the calls to quantum minimum finding, all other operations are deterministic. 
        Therefore, we conclude that the overall success probability of $\mathsf{QDynamicRMQ}$ implemented by \cref{algo:qdynamic-rmq} is at least $1-\varepsilon$.

        \textit{Complexity}. 
        Let $\varepsilon' = \varepsilon/\rbra{4q\log_2 n}$. 
        Then, every call to the quantum minimum finding made in \cref{algo:qdynamic-rmq} is with successful probability at least $1-\epsilon'$.
        We note that $\mathsf{nodes}\sbra{k}$ belongs to layer $\floor{\log_2 k}$, and the range it represents has length $n/2^{\floor{\log_2 k}} \leq 2n / k$. 
        Then, by \cref{lemma:findmin}, calling the $\mathsf{create}\rbra{\mathsf{nodes}\sbra{k}}$ procedure has query complexity $t_k = O\rbra{\sqrt{n/k \cdot \log\rbra{1/\varepsilon'}}}$ and time complexity $O\rbra{t_k\log n}$.
        For top-down completion, if $q$ RMQ operations are processed, then $\mathsf{nodes}\sbra{1}, \mathsf{nodes}\sbra{2}, \dots, \mathsf{nodes}\sbra{q}$ are created, i.e., $\mathsf{create}\rbra{\mathsf{nodes}\sbra{k}}$ is called for each $1 \leq k \leq q$. In this case, the query complexity caused by top-down completion is
        \[
        \mathsf{Q}_1 := \sum_{k=1}^q t_k = \sum_{k=1}^q O\rbra*{\sqrt{\frac n k \log {\frac 1 {\varepsilon'}}}} = O\rbra*{\sqrt{nq \log {\frac 1 {\varepsilon'}}}}.
        \]

        On the other hand, in $k$-th operation, the nodes in layers $\lfloor\log_2 k\rfloor,\ldots,s$ may not have been fully created yet. 
        No matter whether the $k$-th operation is $\mathsf{Query}$ or $\mathsf{Modify}$, there will be at most $4$ nodes created in each layer (see \cref{fig:dynamic} for an illustrative example).
        Let $m_1 < m_2 < \dots < m_{\ell}$ be the nodes that are newly created.
        Then, 
        \begin{enumerate}
            \item $m_1 > k$ because of top-down completion. 
            \item $\ell = O\rbra{\log n}$ because any range will involve at most $O\rbra{\log n}$ nodes.
            \item $m_i \geq 2^{\floor{\rbra{i-1}/8}} m_1$ for $1 \leq i \leq \ell$. This is because there must be a layer between the layer of node $m_{i+8}$ and the layer of node $m_i$ (due to the fact that each layer has at most $4$ nodes visited), which implies that $m_{i+8} \geq 2m_i$.
        \end{enumerate}
        Therefore, the query complexity in the $k$-th operation caused by lazy node creation is at most
        \begin{align*}
            \mathsf{Q}_{2,k} := \sum_{i=1}^{\ell} t_{m_i} 
            & = \sum_{i=1}^{\ell} O\rbra*{\sqrt{\frac n {m_i} \log {\frac 1 {\varepsilon'}}}} \\
            & = \sum_{i=1}^{\ell} O\rbra*{\sqrt{\frac n {2^{\floor{\rbra{i-1}/8}} m_1} \log {\frac 1 {\varepsilon'}}}} \\
            & = O\rbra*{\sqrt{\frac n {m_1} \log {\frac 1 {\varepsilon'}}}}  \\
            & = O\rbra*{\sqrt{\frac n {k} \log {\frac 1 {\varepsilon'}}}}.
        \end{align*}
        In summary, the overall query complexity is 
        \[
        \mathsf{Q} := \mathsf{Q}_1 + \sum_{k = 1}^q \mathsf{Q}_{2, k} = O\rbra*{\sqrt{nq \log {\frac 1 {\varepsilon'}}}} + \sum_{k = 1}^q O\rbra*{\sqrt{\frac n {k} \log {\frac 1 {\varepsilon'}}}} = O\rbra*{\sqrt{nq \log \frac{q \log n}{\varepsilon}}}.
        \]
        A similar analysis will show that the overall time complexity is
        \[
        O\rbra{\mathsf{Q} \log n} = O\rbra*{\sqrt{nq \log \frac{q \log n}{\varepsilon}} \log n}.
        \]

        Finally, we consider the case when the exact value of $q$ is not known but only an upper bound $q_0 \geq q$ is given. 
        In this case, we can just initialize $\mathsf{QDynamicRMQ}$ by calling $\mathsf{Initialize}\rbra{\mathcal{O}_a, q_0, \varepsilon}$. 
        In particular, when $q_0 = n$, the query complexity is $\mathsf{Q} = O\rbra{\sqrt{nq\log\rbra{n/\varepsilon}}}$. 
    \end{proof}

    \subsection{Application: quantum \texorpdfstring{$k$}{k}-minimum finding} \label{sec:k-min-find}
    
    As an application, we obtain a quantum algorithm for $k$-minimum finding with the help of the data structure $\mathsf{QDynamicRMQ}$ given in \cref{thm:QDynamicRMQ}.

    \begin{corollary} [Quantum $k$-minimum finding without QRAM] \label{corollary:qkmin}
        There is a quantum algorithm for $k$-minimum finding with success probability $\geq 1-\varepsilon$ with query complexity $O\rbra{\sqrt{nk\log\rbra{k\log\rbra{n}/\varepsilon}}}$ and time complexity $O\rbra{\sqrt{nk\log\rbra{k\log\rbra{n}/\varepsilon}}\log{n}}$. 
        Moreover, the algorithm only uses $O\rbra{\log{n}}$ qubits and it does not require the assumption of QRAM.
    \end{corollary}

    \begin{algorithm}[!htp]
        \caption{$\mathsf{QFindMin}\rbra{\mathcal{O}_a, k, \varepsilon}$}
        \label{algo:qkmin}
        \begin{algorithmic}[1]
        \State $\mathcal{T} \gets \mathsf{QDynamicRMQ}.\mathsf{Initialize}\rbra{\mathcal{O}_a, 2k, \varepsilon}$;
        \For{$j = 1 \to k$}
            \State $i_j \gets \mathcal{T}.\mathsf{Query}\rbra{1, n}$;
            \State $\mathcal{T}.\mathsf{Modify}\rbra{i_j, i_j, f \colon x \mapsto +\infty}$;
        \EndFor
        \State \Return $\rbra{i_1, i_2, \dots, i_k}$. 
        \end{algorithmic}
    \end{algorithm}
    
    \begin{proof}
        We first explain the idea of finding (the indices of) the $k$ minimum elements of the input array $a\substr{1}{n}$ with a data structure for RMQ. 
        For simplicity, we first assume that every operation on the data structure succeeds with certainty. 
        With the RMQ data structure, we can find the $k$ minimum elements by repeating the following instructions using the $\mathsf{Query}$ and $\mathsf{Modify}$ operations:
        \begin{enumerate}
            \item Initialize the RMQ data structure $\mathcal{T}$ with the array $a\substr{1}{n}$.  
            \item For each $j = 1, 2, \dots, k$,
        \begin{enumerate}
            \item[2.1.] Let $i_j$ be the index of the minimum element, which is obtained by calling $\mathcal{T}.\mathsf{Query}\rbra{1, n}$. 
            \item[2.2.] Set $a\sbra{i_j} \gets +\infty$, which is done by calling $\mathcal{T}.\mathsf{Modify}\rbra{i_j, i_j, f \colon x \mapsto +\infty}$. 
        \end{enumerate}
        \end{enumerate}
        It can be seen by induction that after the first $\ell$ loops, we can find the indices of the $\ell$ smallest elements $i_1, i_2, \dots, i_\ell$ in the array $a\substr{1}{n}$ and then set $a\sbra{i_1}, a\sbra{i_2}, \dots, a\sbra{i_\ell}$ to $+\infty$ (as if they were removed), which guarantees that at the beginning of the $\rbra{\ell+1}$-th loop, the minimum element of the current array holds the $\rbra{\ell+1}$-th smallest element of the original array.

        We extend this idea to the case that operations on the data structure could fail with bounded-error probability. 
        The detailed implementation is given in \cref{algo:qkmin}.

    \textit{Complexity}. The complexity is obtained from \cref{thm:QDynamicRMQ} by setting $q = 2k$. 

    \textit{Correctness}. The parameters in the initialization $\mathsf{QDynamicRMQ}.\mathsf{Initialize}\rbra{\mathcal{O}_a, 2k, \varepsilon}$ guarantee the correctness. 
    \end{proof}
    
    \subsection{Quantum lower bound} \label{sec:lower-bound-dynamic-rmq}

    To conclude this section, we give a matching (up to logarithmic factors) quantum lower bound for RMQ. 

    \begin{theorem} [Lower bound for quantum RMQ] \label{thm:QDynamicRMQ-lower-bound}
        Any quantum algorithm for RMQ that supports $q$ operations on an array of length $n$ has query complexity $\Omega\rbra{\sqrt{nq}}$. 
    \end{theorem}

    \begin{proof}
        The lower bound is based on the reduction from $k$-minimum finding to RMQ in \cref{corollary:qkmin}.
        The construction in \cref{corollary:qkmin} means that any quantum RMQ with query complexity $T\rbra{n, q}$ can be used to solve quantum $k$-minimum finding with query complexity $T\rbra{n, 2k}$. 
        On the other hand, it is known in \cite[Theorem 8.1]{DHHM06} that quantum $k$-minimum finding requires at least $\Omega\rbra{\sqrt{nk}}$ queries, i.e., $T\rbra{n, 2k} = \Omega\rbra{\sqrt{nk}}$. 
        Therefore, due to the arbitrariness of $k$, we conclude that $T\rbra{n, q} = \Omega\rbra{\sqrt{nq}}$. 
    \end{proof}

\section*{Acknowledgements}

The authors would like to thank the anonymous reviewers for noting another approach for quantum $k$-minimum finding (mentioned in \cref{footnote:another-qkmin}) and for raising a possible approach to Question~\ref{ques:1}.

The work of Qisheng Wang was supported in part by the Engineering and Physical Sciences Research Council under Grant \mbox{EP/X026167/1}.
The work of Zhicheng Zhang was supported in part by the Sydney Quantum Academy, NSW, Australia and in part by
the Australian Research Council (Grant Number: DP250102952).

\bibliographystyle{unsrturl}
\bibliography{main}

\end{document}